\documentclass[conference,doublecolumn]{IEEEtran}

\newtheorem{thm}{Theorem}
\newtheorem{cor}{Corollary}
\newtheorem{lem}{Lemma}

\newcommand{\entropy}[1]{\mathsf{h}\left(#1\right)}

\usepackage[final]{graphicx}
\usepackage[reqno]{amsmath}
\usepackage{amssymb}
\usepackage{subfig}
\usepackage{epstopdf}

\begin{document}

\title{Degrees of Freedom (DoF) of Locally Connected Interference Channels with Coordinated Multi-Point (CoMP) Transmission}
\author{{\large{Aly El Gamal, V.~Sreekanth Annapureddy, and Venugopal V.~Veeravalli}}\\ \large{ECE Department and Coordinated Science Laboratory}\\\large{University of Illinois at Urbana-Champaign}}

\maketitle

\begin{abstract}
The degrees of freedom (DoF) available for communication provides an analytically tractable way to characterize the information-theoretic capacity of interference channels.  In this paper, the DoF of a $K$-user interference channel is studied under the assumption that the transmitters can cooperate via coordinated multi-point (CoMP) transmission. In~\cite{wigger-wynerasymm}, the authors considered the linear asymmetric model of Wyner, where each transmitter is connected to its own receiver and its successor, and is aware of its own message as well as $M-1$ preceding messages. The per user DoF was shown to go to $\frac{M}{M+1}$ as the number of users increases to infinity. In this work, the same model of channel connectivity is considered, with a relaxed cooperation constraint that bounds the maximum number of transmitters at which each message can be available, by a {\em cooperation order} $M$. We show that the relaxation of the cooperation constraint, while maintaining the same load imposed on a backhaul link needed to distribute the messages, results in a gain in the DoF. In particular, the asymptotic limit of the per user DoF under the cooperation order constraint is $\frac{2M}{2M+1}$. Moreover, the optimal transmit set selection satisfies a {\em local cooperation} constraint. i.e., each message needs only to be available at neighboring transmitters.  
\end{abstract}

\section{Introduction}
%Justifying the DoF analysis
%fully connected Gaussian IC
As a result of developments in the infrastructure of cellular networks, there has been a recent growing interest in the potential of {\em cooperative} transmission techniques where, through a backhaul link, messages can be available at more than one transmitter, i.e., Coordinated Multi-Point (CoMP) transmission. This new development has a proven advantage~(see, e.g.,~\cite{CoMP-book}) for mitigating the effect of interfering signals, in particular, for cell-edge users.  
%CoMP reduces cell edge IC

The number of degrees of freedom (DoF) available for communication in a given channel is defined as the pre-log factor of its sum capacity. This criterion provides an analytically tractable way to characterize the sum capacity and captures the number of \emph{interference-free} sessions in a given multiuser channel.  In~\cite{Madsen-Nosratinia}, the DoF for the fully connected $K-$user Gaussian interference channel was shown to be upper bounded by half the number of users $K/2$ ($1/2$ per user DoF). This was shown to be achievable through the interference alignment (IA) scheme in~\cite{Cadambe-IA}. However, the IA scheme makes use of symbol extensions, i.e., coding over multiple realizations of the channel, and the $1/2$ per user DoF is achievable only in the asymptotic limit as the number of symbol extensions increases to infinity.  In this work, it will be of interest to compare the DoF and its achieving coding scheme for the fully connected Gaussian interference channel with that of a scenario with two different aspects of practical relevance, namely, CoMP transmission and local connectivity.  
%information theoretic IC --> results for fully connected Gaussian IC
%goal of this work

In~\cite{wigger-wynerasymm}, the authors considered the linear asymmetric channel model which was first introduced by Wyner~\cite{Wyner}, where each transmitter is connected to its own receiver and one successive receiver. For this channel model, they characterized the DoF under a cooperation model that assumes each user's transmitter to be aware of its own message as well as messages belonging to $M-1$ preceding users. The per user DoF was shown to go to $\frac{M}{M+1}$ as $K$ increases to $\infty$. An interesting feature of the achievable scheme in this case is that it does not use symbol extensions. The result of~\cite{wigger-wynerasymm} suggests that the addition of CoMP transmission and local connectivity, not only provides a more realistic model, but leads to a simpler solution.
%advantage of the new aspects
%wigger's work, advantage of locally connected model and coop model

In this work, we relax the message assignment assumption of~\cite{wigger-wynerasymm}, and consider all possible assignments that satisfy a {\em cooperation order} constraint. That is, we limit the maximum number of transmitters at which any given message can be available, or the maximum size of a transmit set, by $M$. Our main result shows that this relaxation leads to an asymptotic limit of the per user DoF for Wyner's linear asymmetric model of $\frac{2M}{2M+1}$, which is strictly higher than that of~\cite{wigger-wynerasymm}. Moreover, the optimal message assignment is shown to satisfy a {\em local cooperation} constraint. i.e., each message needs to be available only at neighboring transmitters, thus retaining the same advantages as the message assignment considered in~\cite{wigger-wynerasymm}.

%Questions of this work

%Document Organization
We provide a precise formulation of the problem in Section~\ref{sec:systemmodel}. In Section~\ref{sec:mainresult}, we state and prove our main result, i.e., a characterization of the limit of the per user DoF in Wyner's linear asymmetric model as the number of users increases. We then discuss the result in Section~\ref{sec:discussion}. Finally, in Section~\ref{sec:conclusion}, we provide concluding remarks.

\section{System Model and Notation}\label{sec:systemmodel}
We use the standard model for the $K-$user interference channel with a single antenna at each node.
\begin{equation*}
Y_i(t) = \sum_{j=1}^{K} h_{ij}(t) X_j(t) + Z_i(t),i\in[K]
\end{equation*}
where $t$ is the time index, $X_i(t)$ is the transmitted signal of transmitter $i$, $Y_i(t)$ is the received signal of receiver $i$, $Z_i(t)$ is the zero mean unit variance Gaussian noise at receiver $i$, $h_{ij} (t)$ is the channel coefficient from transmitter $j$ to receiver $i$ over the $t^{th}$ time slot, and $[K]$ denotes the set $\{1,2,\ldots,K\}$. 

For any set ${\cal A} \subseteq [K]$, we define the complement set $\bar{\cal A} = \{i: i\in[K], i\notin {\cal A}\}$. For each $i \in [K]$, let $W_i$ be the message intended for receiver $i$, we use the abbreviations $W_{\cal A}$, $X_{\cal A}$, and $Y_{\cal A}$ to denote the sets $\{W_i, i\in {\cal A}\}$, $\{X_i, i\in {\cal A}\}$, and $\{Y_i, i\in {\cal A}\}$, respectively.

\subsection{Channel Model}
Each transmitter is connected to its corresponding receiver as well as one following receiver, and the last transmitter is connected only to its corresponding receiver. More precisely,

\begin{equation}\label{eq:channel}
h_{ij} \neq 0 \text { if and only if } i \in \{j,j+1\},\forall i,j \in [K]
\end{equation}

All the channel coefficients are assumed to be fixed and known at all transmitters and receivers. The channel model is illustrated for $K=3$ in Figure~\ref{fig:wynermodel}.

\begin{figure}[htb]
\centering
\includegraphics[width=0.8\columnwidth]{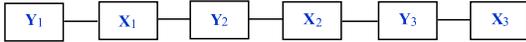}
\caption{Wyner's linear asymmetric model for $K=3$. In the figure, a solid line connects a transmitter-receiver pair if and only if the channel coefficient between them is non-zero.}
\label{fig:wynermodel}
\end{figure}

\subsection{Cooperation Model}
For each $i \in [K]$, let ${\cal T}_i \subseteq [K]$ be the transmit set of message $W_i$. The transmitters in ${\cal T}_i$ cooperatively transmit the message $W_i$ to the receiver $i$. The messages $\{W_i\}$ are assumed to be independent of each other. The \emph{cooperation order} $M$ is defined as the maximum size of a transmit set:
\begin{equation}\label{eq:coop_order}
M = \max_i |{\cal T}_i|.
\end{equation}
\subsection{Degrees of Freedom}
The total power constraint across all the users is $P$.  The rates $R_i(P) = \frac{\log|W_i|}{n}$ are achievable if the error probabilities of all messages can be simultaneously made arbitrarily small for large enough $n$. The capacity region $\mathcal{C}(P)$ is the set of all achievable rate tuples. The DoF ($\eta$) is defined as $\limsup_{P \rightarrow \infty}\frac{ C_{\Sigma}(P)}{\log P}$, where $C_\Sigma(P)$ is the sum capacity. Since $\eta$ depends on the specific choice of transmit sets, we define $\eta(K,M)$ as the best achievable $\eta$ over all choices of transmit sets satisfying the cooperation order constraint in \eqref{eq:coop_order} for a $K-$user channel satisfying~\eqref{eq:channel}.  We define the per user DoF $\tau(M)$ to measure how the sum degrees of freedom scales with the number of users for a fixed cooperation order.
\begin{equation}
\tau(M) = \lim_{K\rightarrow \infty} \frac{\eta(K,M)}{K}
\end{equation}

It is worth noting here that modifying the channel model such that the channel coefficient between the last transmitter and first receiver ($h_{1K}$) is non-zero (cyclic model) does not change the value of $\tau(M)$.

\section{Example: $M=1$}\label{sec:example}
In~\cite{wigger-wynerasymm}, it is shown that the per user DoF equals $\frac{M}{M+1}$  if we fix the transmit sets to be of the form.
\begin{equation}\label{eq:wigger-txset}
{\cal T}_i = \{i, i+1, \ldots, i+M-1\}
\end{equation}

Before coming to our main result, we provide a simple example that motivates our revisit of the work in~\cite{wigger-wynerasymm} with the cooperation order constraint~\eqref{eq:coop_order}. Consider the case where CoMP transmission is not allowed. i.e., $M=1$. Assuming that each message is only available at its corresponding transmitter, then the DoF for this channel can be shown to be $\frac{K}{2}$ where $K$ is even. Now, we relax the constraint that sets each message only at its own transmitter, to another that allows each message to be available at only one transmitter - without specifying that transmitter. Let $W_1, W_3$, be available at $X_1, X_2$, respectively, and deactivate both the second receiver $Y_2$ and the third transmitter $X_3$, then it is easily seen that messages $W_1$ and $W_3$ can be received without interfering signals at their corresponding receivers. Moreover, the deactivation of $X_3$ splits this part of the network from the rest. i.e., the same scheme can be repeated by assigning $W_4, W_6$, to $X_4, X_5$, respectively, and so on. Thus, $2$ degrees of freedom can be achieved for each set of $3$ users, thereby, achieving $\frac{2K}{3}$ DoF where $K$ is a multiple of $3$. The described message assignment is depicted in Figure~\ref{fig:examplemone}. It is evident now that a constraint that is only a function of the load on the {\em backhaul} link may lead to a discovery of better message assignments than the one considered in~\cite{wigger-wynerasymm}. In Section~\ref{sec:mainresult}, we show that the optimal message assignment under the  cooperation order constraint~\eqref{eq:coop_order} is different from the one defined in~\eqref{eq:wigger-txset} .

\begin{figure}[htb]
\centering
\includegraphics[width=0.8\columnwidth]{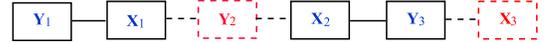}
\caption{Achieving $2/3$ per user DoF for $M=1$. Each transmitter is carrying a message for the receiver connected to it by a solid line. Figure showing only signals corresponding to the first $3$ users in a general $K-$user network. Signals in dashed boxes are deactivated. Note that the deactivation of $X_3$ splits this part of the network from the rest.}
\label{fig:examplemone}
\end{figure}

\section{Useful Message Assignments}\label{sec:usefulmsgdistributions}
In order to characterize the per user DoF $\tau(M)$, we have to consider all possible message assignments satisfying the cooperation order constraint~\eqref{eq:coop_order}. In this section, we characterize necessary conditions for the optimal message assignment. The constraints we provide for transmit sets are governed by the connectivity pattern of the channel. For example, for the case where $M=1$, any assignment of message $W_i$ to a transmitter that is not connected to $Y_i$ is not \emph{useful}, i.e., the rate of transmitting message $W_i$ has to be zero for those assignments. 

For message $W_i$, and a fixed transmit set ${\cal T}_i$, we construct the following graph $G_{W_i,{\cal T}_i}$ that has $[K]$ as its set of vertices, and an edge exists between any given pair of vertices $x,y \in [K]$ if and only if,
\begin{itemize}
\item  $x,y \in {\cal T}_i$. 
\item Corresponding transmitters are both connected to at least one receiver. In the channel model defined in~\eqref{eq:channel}, this condition reduces to $(x-y) \in \{-1,1\}$.
\end{itemize}

Vertices corresponding to transmitters connected to $Y_i$ are given a special mark, i.e., vertices with labels $i$ and $i-1$ are marked for the considered channel model. In Figure~\ref{fig:usefulmsgassignment}, we give an example for the construction of $G_{W_i,{\cal T}_i}$ and the application of the following lemma.

\begin{lem}\label{lem:usefulmsgdistributions}
For any $k \in {\cal T}_i$ such that the vertex $k$ in $G_{W_i,{\cal T}_i}$ is not connected to a marked vertex, removing $k$ from ${\cal T}_i$ does not decrease the sum rate.
\begin{proof}
Let ${\cal S}$ denote the set of indices of vertices in a component with no marked vertices, ${\cal S}'$ be the set of indices of received signals that are connected to at least one transmitter with an index in ${\cal S}$. To prove the lemma, we consider two scenarios, where we add a {\em tilde} over symbols denoting rates and signals belonging to the second scenario. For the first scenario, $W_i$ is made available at transmitters in ${\cal S}$. Let $Q$ be a random variable that is independent of all messages and has the same distribution as $W_i$, then for the second scenario, $W_i$ is not available at transmitters in ${\cal S}$, and a realization $q$ of $Q$ is generated and given to all nodes in $\tilde{X}_{\cal S} \cup \tilde{Y}_{\cal S'}$ before communication starts. Moreover, the given realization $Q=q$ contributes to the encoding of $\tilde{X}_{\cal S}$ in the same fashion as a message $W_i=q$ contributes to $X_{\cal S}$. Assuming a reliable communication scheme for the first scenario that uses a large block length $n$, the following argument shows that the achievable sum rate is also achievable after removing $W_i$ from the designated transmitters.
\begin{eqnarray*}
n\sum_j R_j &=& \sum_j H(W_j) \notag
\\&\overset{(a)}{\leq}& \sum_j I(W_j;Y_j) + o(n) \notag
\\&=& \sum_{j \in {\cal S}'^c} I(W_j,Y_j) + \sum_{j \in {\cal S}'} I(W_j;Y_j)+ o(n) \notag
\\&\overset{(b)}{=}& \sum_{j \in {\cal S}'^c} I(W_j,\tilde{Y}_j) + \sum_{j \in {\cal S}'} I(W_j;Y_j)+ o(n) \notag
\\&\leq&\sum_{j \in {\cal S}'^c} I(W_j,\tilde{Y}_j) + \sum_{j \in {\cal S}'} I(W_j;Y_j|W_i)+ o(n) \notag
\\&\overset{(c)}{=}&\sum_{j \in {\cal S}'^c} I(W_j,\tilde{Y}_j) + \sum_{j \in {\cal S}'} I(W_j;\tilde{Y}_j|Q)+ o(n) \notag
\\&\leq& n\sum_j \tilde{R}_j + o(n)
\end{eqnarray*}
where $(a)$ follows from Fano's inequality, $(b)$ follows as the difference between the two scenarios lies in the encoding of $X_{\cal S}$ which affects only $Y_{{\cal S}'}$, and $(c)$ follows as there are no transmitters outside $X_{\cal S}$ that are carrying $W_i$ and connected to $Y_{{\cal S}'}$, and $Y_i \notin Y_{{\cal S}'}$.
\end{proof}
\end{lem}
We call a message assignment \emph{useful} if no element in it can be removed without decreasing the sum rate. The following corollary to the above lemma characterizes a necessary condition for any message assignment satisfying the cooperation order constraint in~\eqref{eq:coop_order} to be useful.
\begin{cor}\label{cor:usefulmsgdistribution}
Let ${\cal T}_i$ be a useful message assignment and $|{\cal T}_i| \leq M$, then $\forall k\in[K], k\in{\cal T}_i$ only if the vertex $k$ in $G_{W_i,{\cal T}_i}$ lies at a distance that is less than or equal $M-1$ from a marked vertex.
\end{cor}

\begin{figure}[htb]
\centering
\includegraphics[width=0.8\columnwidth]{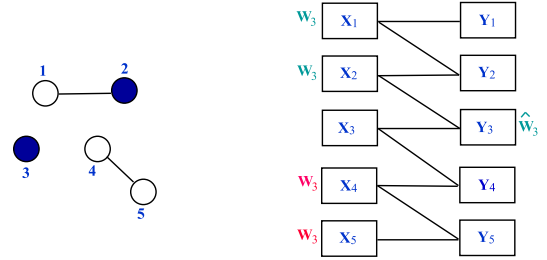}
\caption{Figure showing the construction of $G_{W_3,{\cal T}_3}$ in a $5-$user channel. Marked vertices are represented with filled circles. $W_3$ can be removed at both $X_4$ and $X_5$ without decreasing the sum rate, as the corresponding vertices lie in a component that does not contain a marked vertex.}
\label{fig:usefulmsgassignment}
\end{figure}

\section{Main Result}\label{sec:mainresult}\label{sec:mainresult}

Here, we provide an asymptotic characterization of $\tau(M)$. It is worth noting that the following result shows that the per user DoF under the general cooperation order constraint~\eqref{eq:coop_order} is strictly greater than the $\frac{M}{M+1}$ per user DoF shown in~\cite{wigger-wynerasymm} for the message assignment defined in~\eqref{eq:wigger-txset}.

\begin{thm}\label{thm:Asymmetric Model}
\begin{equation}
\tau(M) = \frac{2M}{2M+1}
\end{equation}
\end{thm}

{\bf Coding Scheme:}
We treat the network as a set of clusters, each consisting of consecutive $2M+1$ transceivers. The last transmitter of each cluster is deactivated to eliminate {\em inter-cluster} interference. It then suffices to show that $2M$ DoF can be achieved in each cluster. Without loss of generality, consider the cluster with users of indices in the set $[2M+1]$. We define the following subsets of $[2M+1]$,
\begin{eqnarray*}
{\cal S}_1 &=& [M]
\\{\cal S}_2 &=& \{M+2,M+3,\ldots,2M+1\}
\end{eqnarray*}
We next show that each user in ${\cal S}_1 \cup {\cal S}_2$ achieves one degree of freedom while message $W_{M+1}$ is not transmitted. In the proposed scheme, users in the set ${\cal S}_1$ are served by transmitters in the set $\{X_1,X_2,\ldots,X_M\}$ and users in the set ${\cal S}_2$ are served by transmitters in the set $\{X_{M+1},X_{M+2},\ldots,X_{2M}\}$. Let the message assignments be as follows.\\

${\cal T}_{i}=
\begin{cases}
\{i,i+1,\ldots,M\}, \quad &\forall i \in {\cal S}_1\\
\{i-1,i-2,\ldots,M+1\},\quad &\forall  i \in {\cal S}_2
\end{cases}$\\

Now, due to the availability of channel state information at the transmitters, the transmit beams for message $W_i$ can be designed to cancel its effect at receivers with indices in the set ${\cal C}_i$, where,\\

${\cal C}_{i}=
\begin{cases}
\{i+1, i+2, \ldots,M\},\quad  &\forall i \in {\cal S}_1\\
\{i-1, i-2, \ldots,M+2\},\quad &\forall  i \in {\cal S}_2
\end{cases}$\\

%Figure illustrating the coding scheme for M=3
Note that both ${\cal C}_M$ and ${\cal C}_{M+2}$ equal the empty set, as both $W_M$ and $W_{M+2}$ do not contribute to interfering signals at receivers in the set $Y_{{\cal S}_1} \cup Y_{{\cal S}_2}$. The above scheme for $M=3$ is illustrated in Figure~\ref{fig:mthree} (a).
We conclude that each receiver with index in the set ${\cal S}_1\cup{\cal S}_2$ suffers only from Gaussian noise, thereby enjoying one degree of freedom.

It is worth noting that the proposed coding scheme is similar to that of~\cite{Shamai-Wigger-ISIT11}. More specifically, the scheme suggested in~(\cite{Shamai-Wigger-ISIT11}, Remark $2$) can be used to achieve a per user DoF of $\frac{2M-1}{2M}$. However, the proposed scheme achieves a higher value as we do not insist on assigning each message to the transmitter with the same index. 

{\bf Outer Bound}
We first provide a general lemma that is used for upper bounding the available degrees of freedom for reliable communication over a multiuser channel. Converse proofs that use a similar argument as the following lemma exist in the literature (see, e.g.,~\cite{Shamai-Wigger-ISIT09}).

For any set ${\cal A} \subseteq [K]$, Define $U_{\cal A} = \cup_{i \notin {\cal A}} {\cal T}_i$. Assume that $\eta$ degrees of freedom are available for the considered channel.
\begin{lem}\label{thm:dofouterbound}
If there exists a set ${\cal A}\subseteq [K]$ and a function $f$, such that $f\left(Y_{\cal A},Z_{\cal A},X_{\bar{U}_{\cal A}}\right)=X_{U_{\cal A}}$, then $\eta \leq |{\cal A}|$. 
\end{lem} 
\begin{proof}
Proof is available in Appendix.
\end{proof}

In order to prove the converse, we use Lemma~\ref{thm:dofouterbound} with a set ${\cal A}$ of size $K\frac{2M}{2M+1}+o(K)$. We also prove the upper bound for the channel after removing the first $M$ transmitters $\left(X_{[M]}\right)$, while noting that this will be a valid bound on $\tau(M)$ since the number of removed transmitters is $o(K)$. 

Inspired by the coding scheme, we define the set ${\cal A}$ as the set of receivers that are {\em active} in the above described strategy. i.e., the complement set ${\bar{\cal A}}=\{i: i\in[K], i= (2M+1)(j-1)+M+1, j \in {\bf{Z}^+}\}$. We know from Corollary~\ref{cor:usefulmsgdistribution} that messages belonging to the set $W_{\bar{\cal A}}$ do not contribute to transmit signals with indices that are multiples of $2M+1$, i.e., $i \notin U_{\cal A}$ for all $i\in[K]$ that is a multiple of $2M+1$. More precisely, let the set ${\cal S}$ be defined as follows:
\begin{equation*}
{\cal S} = \{i: i\in[K], i \text{ is a multiple of } 2M+1\}
\end{equation*}
then ${\cal S}\subseteq {\bar{U}_{\cal A}}$. In particular, $X_{\cal S} \subseteq X_{\bar{U}_{\cal A}}$, hence it suffices to show the existence of a function $f$ such that $f\left(Y_{\cal A},Z_{\cal A},X_{\cal S}\right)=X_{\bar{S}} \backslash X_{[M]}$. We show in what follows how to reconstruct the signals in the set $\{X_{M+1},X_{M+2},\ldots,X_{2M}\}\cup\{X_{2M+2},X_{2M+3},\ldots,X_{3M+1}\}$, then it will be clear by symmetry how to reconstruct the rest of transmit signals in the set $X_{\bar{\cal S}}\backslash X_{[M]}$. Since $X_{2M+1} \in X_{\cal S}$, and a noise free version of ${Y_{2M+1}}$ is also given, then $X_{2M}$ can be reconstructed. Now, with the knowledge of $X_{2M}$, $Y_{2M}$, and $Z_{2M}$, we can reconstruct $X_{2M-1}$, and so by iterative processing all transmit signals in the set $\{X_{M+1},X_{M+2},\ldots,X_{2M}\}$ can be reconstructed. In a similar fashion, given $X_{2M+1}$, $Y_{2M+2}$, and $Z_{2M+2}$, the signal $X_{2M+2}$ can be reconstructed, then with a noise free version of $Y_{2M+3}$, we can reconstruct $X_{2M+3}$, and we can proceed along this path to reconstruct all transmit signals in the set $\{X_{2M+2},X_{2M+3},\ldots,X_{3M+1}\}$. In Figure~\ref{fig:mthree} (b), we illustrate how the proof works for the case where $M=3$. This proves the existence of the function $f$ defined above, and so by Lemma~\ref{thm:dofouterbound} we obtain the converse of Theorem~\ref{thm:Asymmetric Model}.

\begin{figure}
  \centering
  
\subfloat[]{\label{fig:mthreejonecs}\includegraphics[height=0.283\textwidth]{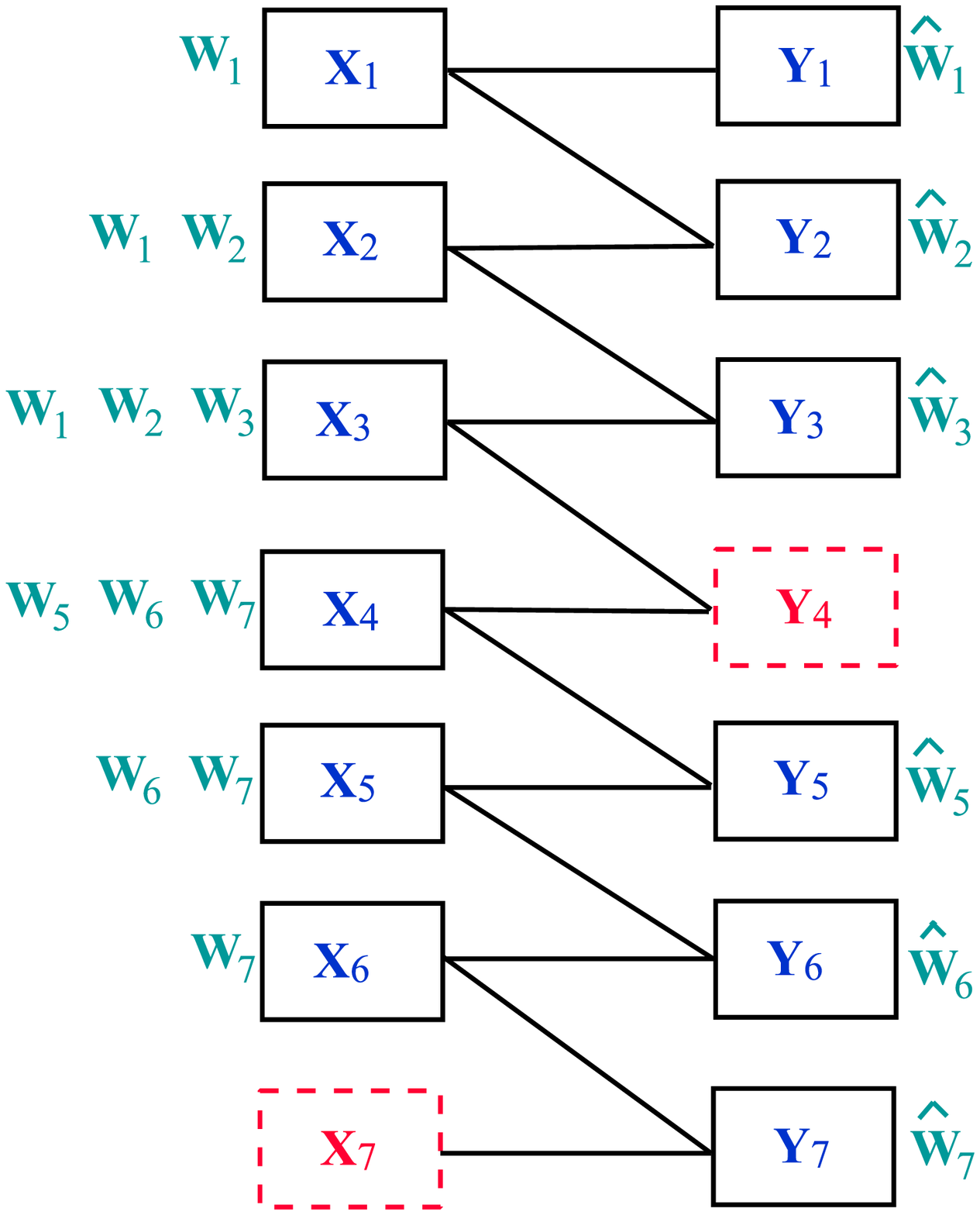}}                
\quad\quad\quad\quad\subfloat[]{\label{fig:mthreejoneub}\includegraphics[width=0.17\textwidth]{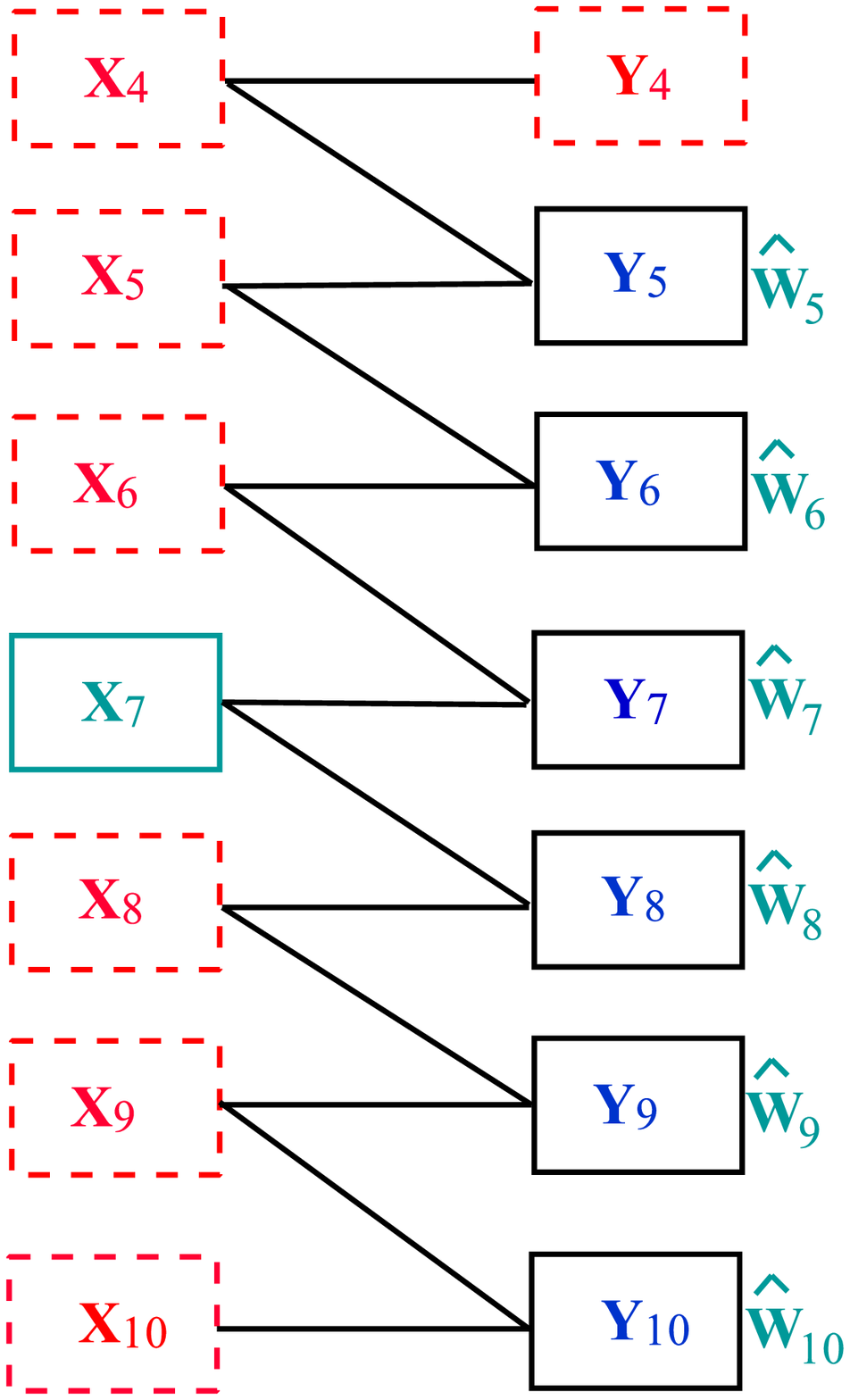}}
  \caption{Figure illustrating the proof of Theorem~\ref{thm:Asymmetric Model} for $M=3$, $\tau(3)=\frac{6}{7}$. In ($a$), the message assignments in the first cluster for the proposed coding scheme are illustrated. Note that both $X_7$ and $Y_4$ are deactivated. In ($b$), an illustration of the upper bound is shown. The messages $W_4$ and $W_{11}$ cannot be available at $X_7$, hence it can be reconstructed from $W_{\cal A}$. All transmit signals shown in figure can be reconstructed from $X_7$ and noise free versions of $\{Y_5,\ldots,Y_{10}\}$}
  \label{fig:mthree}
\end{figure}

\section{Discussion}\label{sec:discussion}
\subsection{Practical Simplicity of the Coding Scheme}
Unlike the fully connected Gaussian interference channel where the DoF of the channel cannot be achieved through linear precoding strategies over finitely many independent channel realizations (see~\cite{Cadambe-IA},~\cite{Bresler-Tse-DoF}), the proposed DoF achieving coding scheme for Wyner's linear model does not involve coding over multiple channel realizations. 

Another aspect justifying the practical simplicity of the proposed coding scheme is the fact that the employed message assignments satisfy a \emph{local cooperation} constraint, that is, regardless of the number of users, message $W_i$ can only be assigned to transmitters whose indices lie within a fixed radius from the index $i$. More precisely, let $r$ be the fixed radius then,
\begin{equation}\label{eq:local_coop}
{\cal T}_i \subseteq [i-r,i+r]
\end{equation}
 We can clearly see that the proposed coding scheme satisfies the constraint in~\eqref{eq:local_coop} with $r=M$, regardless of the number of users $K$.    

\subsection{Extension to Locally Connected Channels}
Consider a symmetric generalization of Wyner's model where each transmitter is connected to $L/2$ following receivers and $L/2$ preceding receivers. More precisely, the channel model is given by,

\begin{equation}\label{eq:generalized_channel}
h_{ij} \neq 0 \text { if and only if } j \in \left[i- \left \lfloor \frac{L}{2} \right \rfloor , i+ \left \lceil \frac{L}{2} \right \rceil \right]
\end{equation}

We note that $L$ is the number of interfering signals at each receiver and for $L=1$ and $L=2$, the channel reduces the commonly known Wyner's asymmetric and symmetric linear models, respectively. It can be shown that the proposed coding scheme in Section~\ref{sec:mainresult} can be generalized to prove that $\tau(M) \geq \frac{2M}{2M+L}$. 

%Extension of coding scheme to a symmetric generalization of wyner's model, and optimality within the class of coding schemes with ZF transmit beams	

\section{Conclusion}\label{sec:conclusion}
In this paper, we characterized the degrees of freedom of Wyner's linear asymmetric model under a cooperation constraint that bounds the size of the transmit sets. The per user DoF was shown to be strictly greater than that achievable by the transmit sets previously considered in~\cite{wigger-wynerasymm}. The proposed coding scheme is simple from a practical viewpoint as it uses zero forcing transmit beams and needs only a single channel realization to achieve the DoF of the channel. Moreover, the proposed message assignment satisfies a local cooperation constraint, where each message can only be available at neighboring transmitters.
\appendix
\section*{Proof of Lemma~\ref{thm:dofouterbound}}
In order to prove the lemma, we show that using a reliable communication scheme with the aid of a signal that is within $o(\log P)$, all the messages can be recovered from the set of received signals $Y_{\cal A}$. It follows that any achievable degree of freedom for the channel is also achievable for another channel that has only those receivers, thus proving the upper bound. 

In any reliable $n$-block coding scheme, \[H(W_i|Y_i^n) \leq n\epsilon, \forall i \in [K].\]
Therefore, \[H(W_{\cal A}|Y_{\cal A}^n) \leq \sum_{i \in {\cal A}} H(W_i|Y_i^n) \leq n|{\cal A}|\epsilon.\]
Now, the sum $\sum_{i \in [K]} R_i = \sum_{i \in \bar{{\cal A}}} R_i + \sum_{i \in {\cal A}}R_i$ can be bounded as
\begin{eqnarray}\label{eq:lemma_tmp1} 
n\left(\sum_{i \in \bar{\cal A}} R_i + \sum_{i \in {\cal A}}R_i\right) & = & H(W_{\bar{\cal A}}) + H(W_{\cal A})\notag  \\
& \leq & I\left(W_{\bar{\cal A}};Y_{\bar{\cal A}}^n\right) + I\left(W_{\cal A};Y_{\cal A}^n\right)\notag\\&&+nK \epsilon. 
\end{eqnarray}
where $\epsilon$ can be made arbitrarily small, by choosing $n$ large enough. The two terms on the right hand side of \eqref{eq:lemma_tmp1} can be bounded as
\begin{eqnarray*}
I\left(W_{\cal A};Y_{\cal A}^n\right) 
& = &  \entropy{Y_{\cal A}^n} - \entropy{Y_{\cal A}^n|W_{\cal A}}\\
& \leq &  \sum_{i \in {\cal A}}\sum_{t = 1}^n \entropy{Y_i(t)} - \entropy{Y_{\cal A}^n|W_{\cal A}} \\
& = &  |{\cal A}|n\log P + n(o(\log P)) - \entropy{Y_{\cal A}^n|W_{\cal A}}
\end{eqnarray*}
\begin{eqnarray*}
I\left(W_{\bar{\cal A}};Y_{\bar{\cal A}}^n\right) & \leq & \ I\left(W_{\bar{\cal A}};Y_{\bar{\cal A}}^n,Y_{\cal A}^n,W_{\cal A}\right) \\
& = &  I(W_{\bar{\cal A}};Y_{\cal A}^n|W_{\cal A}) + I(W_{\bar{\cal A}};Y_{\bar{\cal A}}^n|W_{\cal A},Y_{\cal A}^n) \\
& = &  \entropy{Y_{\cal A}^n|W_{\cal A}}  - \entropy{Z_{\cal A}^n} + \entropy{Y_{\bar{\cal A}}^n|W_{\cal A},Y_{\cal A}^n}\\ &&- \entropy{Z_{\bar{\cal A}}^n}.
\end{eqnarray*}
Now, we have
\begin{eqnarray*}
I\left(W_{\cal A};Y_{\cal A}^n\right)  + I\left(W_{\bar{\cal A}};Y_{\bar{\cal A}}^n\right) &\leq&
  |{\cal A}|n\log P  + \entropy{Y_{\bar{\cal A}}^n|W_{\cal A},Y_{\cal A}^n}\\&&  + n(o(\log P)).
\end{eqnarray*}
Therefore, if we show that \[\entropy{Y_{\bar{\cal A}}^n|W_{\cal A},Y_{\cal A}^n} = n(o(\log P)),\] then from \eqref{eq:lemma_tmp1}, we have the required outer bound. Since $W_{\cal A}$ contains all the messages carried by transmitters in $X_{\bar{U}_{\cal A}}$, they determine those input signals for the $n$ channel uses. Therefore,
\begin{eqnarray*}
\entropy{Y_{\bar{\cal A}}^n|W_{\cal A},Y_{\cal A}^n}  &=&  \entropy{Y_{\bar{\cal A}}^n|W_{\cal A},Y_{\cal A}^n,X_{\bar{U}_{\cal A}}^n} \\
&\leq& \entropy{Y_{\bar{\cal A}}^n|Y_{\cal A}^n,X_{\bar{U}_{\cal A}}^n} \\
& \leq & \sum_{t = 1}^{n} \entropy{Y_{\bar{\cal A}}(t)|Y_{\cal A}(t),X_{\bar{U}_{\cal A}}(t)} \\
&\leq& \sum_{t = 1}^{n} \entropy{Y_{\bar{\cal A}}(t), Z_{[K]}(t)|Y_{\cal A}(t),X_{\bar{U}_{\cal A}}(t)}
\\&=& \sum_{t = 1}^{n} \entropy{Y_{\bar{\cal A}}(t),Z_{[K]}(t),Y_{\cal A}(t),X_{\bar{U}_{\cal A}}(t)}\\&& - \entropy{Y_{\cal A}(t),X_{\bar{U}_{\cal A}}(t)}
\\&\overset{(a)}{=}&\sum_{t = 1}^{n} \entropy{Z_{[K]}(t),Y_{\cal A}(t),X_{\bar{U}_{\cal A}}(t)}\\&& - \entropy{Y_{\cal A}(t),X_{\bar{U}_{\cal A}}(t)}
\\&=&\sum_{t = 1}^{n}\entropy{Z_{[K]}(t)|Y_{\cal A}(t),X_{\bar{U}_{\cal A}}(t)}
\\&=& n(o(\log P))
\end{eqnarray*}
where $(a)$ follows from the existence of the function $f$ by the statement of the Lemma, as given $Y_{\cal A}, X_{\bar{U}_{\cal A}},Z_{\cal A}$, then $X_{U_{\cal A}}$ can be recovered, and hence $Y_{\bar{\cal A}}$, as $Z_{\bar{\cal A}}$ is given.

\end{document}